\documentclass[sigconf]{acmart}

\AtBeginDocument{%
  \providecommand\BibTeX{{%
    \normalfont B\kern-0.5em{\scshape i\kern-0.25em b}\kern-0.8em\TeX}}}

\copyrightyear{2024}
\acmYear{2024}
\setcopyright{acmlicensed}\acmConference[RecSys '24]{18th ACM Conference on Recommender Systems}{October 14--18, 2024}{Bari, Italy}
\acmBooktitle{18th ACM Conference on Recommender Systems (RecSys '24), October 14--18, 2024, Bari, Italy}
\acmDOI{10.1145/3640457.3688147}
\acmISBN{979-8-4007-0505-2/24/10}

\usepackage{algorithm}
\usepackage{algorithmicx}
\usepackage{algpseudocode}
\usepackage{bbm}
\usepackage{float}
\usepackage{graphicx}
\usepackage[misc]{ifsym}
\usepackage{multirow}
\usepackage{xcolor}
\usepackage{subfigure}
\usepackage{colortbl}
\definecolor{mygray}{gray}{.9}
\usepackage{makecell}
\newtheorem{proposition}{Proposition}

\begin{document}

\title{End-to-End Cost-Effective Incentive Recommendation under Budget Constraint with Uplift Modeling}

\author{Zexu Sun}
\affiliation{%
  \institution{Gaoling School of Artificial Intelligence, Renmin University of China}
  \city{Beijing}
  \country{China}
}
\email{sunzexu21@ruc.edu.cn}

\author{ Hao Yang}
\affiliation{%
  \institution{Gaoling School of Artificial Intelligence, Renmin University of China}
  \city{Beijing}
  \country{China}
}
\email{yanghao@ruc.edu.cn}

\author{Dugang Liu}
\authornote{Corresponding Authors}
\affiliation{%
  \institution{Guangdong Laboratory of Artificial Intelligence and Digital Economy (SZ)}
  \city{Shenzhen}
  \country{China}}
\email{dugang.ldg@gmail.com}

\author{Yunpeng Weng}
\affiliation{%
  \institution{FiT, Tencent}
  \city{Shenzhen}
  \country{China}
}
\email{edwinweng@tencent.com}

\author{Xing Tang}
\authornotemark[1]
\affiliation{%
  \institution{FiT, Tencent}
  \city{Shenzhen}
  \country{China}
}
\email{shawntang@tencent.com}

\author{Xiuqiang He}
\affiliation{%
  \institution{FiT, Tencent}
  \city{Shenzhen}
  \country{China}
}
\email{xiuqianghe@tencent.com}

\renewcommand{\shortauthors}{Zexu Sun, et al.}

\begin{abstract}
In modern online platforms, incentives (\textit{e.g.}, discounts, bonus) are essential factors that enhance user engagement and increase platform revenue. Over recent years, uplift modeling has been introduced as a strategic approach to assign incentives to individual customers.
Especially in many real-world applications, online platforms can only incentivize customers with specific budget constraints.
This problem can be reformulated as the multi-choice knapsack problem (MCKP). 
The objective of this optimization is to select the optimal incentive for each customer to maximize the return on investment (ROI).
Recent works in this field frequently tackle the budget allocation problem using a two-stage approach. 
However, this solution is confronted with the following challenges:
(1) The causal inference methods often ignore the domain knowledge in online marketing, where the expected response curve of a customer should be monotonic and smooth as the incentive increases. 
(2) There is an optimality gap between the two stages, resulting in inferior sub-optimal allocation performance due to the loss of the incentive recommendation information for the uplift prediction under the limited budget constraint.
To address these challenges, we propose a novel \underline{E}nd-to-\underline{E}nd Cost-\underline{E}ffective \underline{I}ncentive \underline{R}ecommendation (E$^3$IR) model under the budget constraint. Specifically, our methods consist of two modules, i.e., the uplift prediction module and the differentiable allocation module.
In the uplift prediction module, we construct prediction heads to capture the incremental improvement between adjacent treatments with the marketing domain constraints (\textit{i.e.}, monotonic and smooth). 
We incorporate integer linear programming (ILP) as a differentiable layer input in the differentiable allocation module. 
Furthermore, we conduct extensive experiments on public and real product datasets, demonstrating that our E$^3$IR improves allocation performance compared to existing two-stage approaches.
\end{abstract}




\maketitle

\section{Introduction}
With the development of online platforms, online marketing has become increasingly essential and competitive~\cite{liu2023explicit}. 
Assigning customer discounts or bonuses is a critical strategy for promoting user conversion and increasing revenue. 
For instance, Taobao utilizes coupons to enhance user activity~\cite{li2020spending}, Booking employs promotions to improve user satisfaction~\cite{albert2022commerce}, and Meituan uses cash bonuses to stimulate user retention~\cite{wang2023multi}. 
In an online marketing scenario, recommending a more significant incentive to customers can increase the purchase probability. 
It's reasonable to hypothesize that an increased incentive is unlikely to change spending behavior significantly.
However, the incentive recommendation is often constrained by a limited budget, which means that only a portion of individuals can receive incentives. 
Therefore, the main challenge is to convert more users and generate higher revenue within the budget constraints. 
This problem is typically formulated as a budget allocation problem in online marketing~\cite{zou2020heterogeneous}.

Several recent studies have tackled this problem using a two-stage approach~\cite{zhao2019uplift,albert2022commerce,zhou2023direct}. 
In the first stage, causal inference methods estimate uplift, and in the second stage, an integer programming formulation finds the optimal allocation based on the predicted uplift. 
Zhao et al.~\cite{zhao2019unified} employ a logit response model~\cite{phillips2021pricing} to forecast treatment effects and subsequently determine the optimal allocation by utilizing root-finding to satisfy the Karush-Kuhn-Tucker (KKT) conditions.
Tu et al.~\cite{tu2021personalized} introduce several advanced estimators, including Causal Tree~\cite{darondeau1989causal}, Causal Forest~\cite{athey2019estimating}, and Meta-Learners~\cite{kunzel2019metalearners}, to estimate the heterogeneous effects. They also regard the second stage as an optimization problem.
There are also some works that utilize policy to address the problem of budget allocation~\cite{xiao2019model,zou2020heterogeneous,zhang2021bcorle,zhou2023direct}.  
Xiao et al.~\cite{xiao2019model} and Zhang et al.~\cite{zhang2021bcorle} have developed reinforcement learning solutions utilizing constrained Markov decision processes to learn an optimal policy directly.
Zhou et al.~\cite{zhou2023direct} expand on the concept of decision-focused learning to accommodate multi-treatments and devise a loss function for learning decision factors for MCKP solutions.

However, there are still some limitations to these methods. 
For the two-stage methods, 
firstly, the existing causal inference methods lack interpretability and do not conform to domain knowledge, which may result in unreliable predictions in practical scenarios. 
As shown in Figure~\ref{fig:example}, the response curve of a user should be monotonic and smooth. 
For example, customers who use coupons of varying values typically yield distinct transaction revenues, with larger coupons generally resulting in higher revenues. 
Thus, without incorporating the marketing domain knowledge, we may get a wrong uplift prediction of the user, which results in a non-optimal incentive recommendation;
Secondly, there is an optimality gap between the two stages, as their objectives are fundamentally different. 
For example, with a significantly low total budget, an effective allocation algorithm should distribute incentives solely among a small subset of users who exhibit high sensitivity to incentives. 
In such cases, a well-trained uplift model might exhibit worse marketing effectiveness if the model demonstrates higher precision across all users on average but performs inadequately for this subset of users.
For the policy-based methods~\cite{lopez2020cost}, learning complex policies solely through a model-free black-box approach, without exploiting causal information, may lead to sample inefficiency.

To solve the limitations above, in this paper, we propose an innovative \underline{E}nd-to-\underline{E}nd Cost-\underline{E}ffective \underline{I}ncentives \underline{R}ecommendation (E$^3$IR) model to address these limitations. 
Our model consists of two modules: the uplift prediction module and the differentiable allocation module. 
In particular, the uplift prediction module employs a reparameterization-based multi-head structure for response prediction, ensuring monotonicity and smoothness by constraining the output to be non-negative and sharing model parameters between incremental improvement prediction heads. 
To mitigate the optimality gap, the differentiable allocation module tackles the budget allocation problem from an end-to-end perspective using a backward pass for the Integer Linear Programming (ILP) problem. 
This module can learn both value terms and constraints of the ILP, enabling universal combinatorial expressivity. 
Additionally, we conduct extensive experiments on various datasets to evaluate the effectiveness of our E$^3$IR model.
\begin{figure}[htbp]
    \centering
    \subfigure[non-monotonic]{\includegraphics[width=0.46\linewidth]{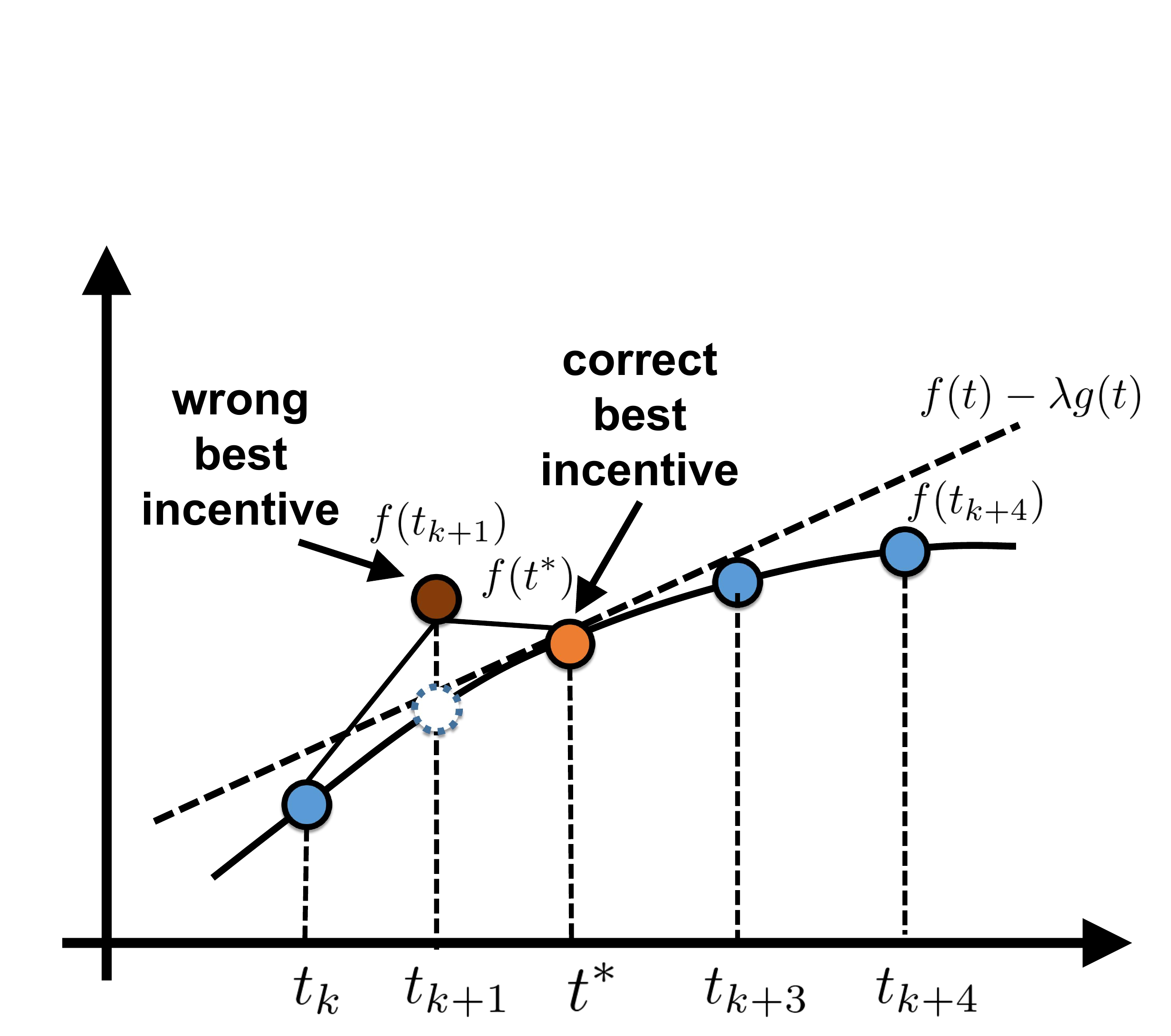}}
    \subfigure[non-smooth]{\includegraphics[width=0.48\linewidth]{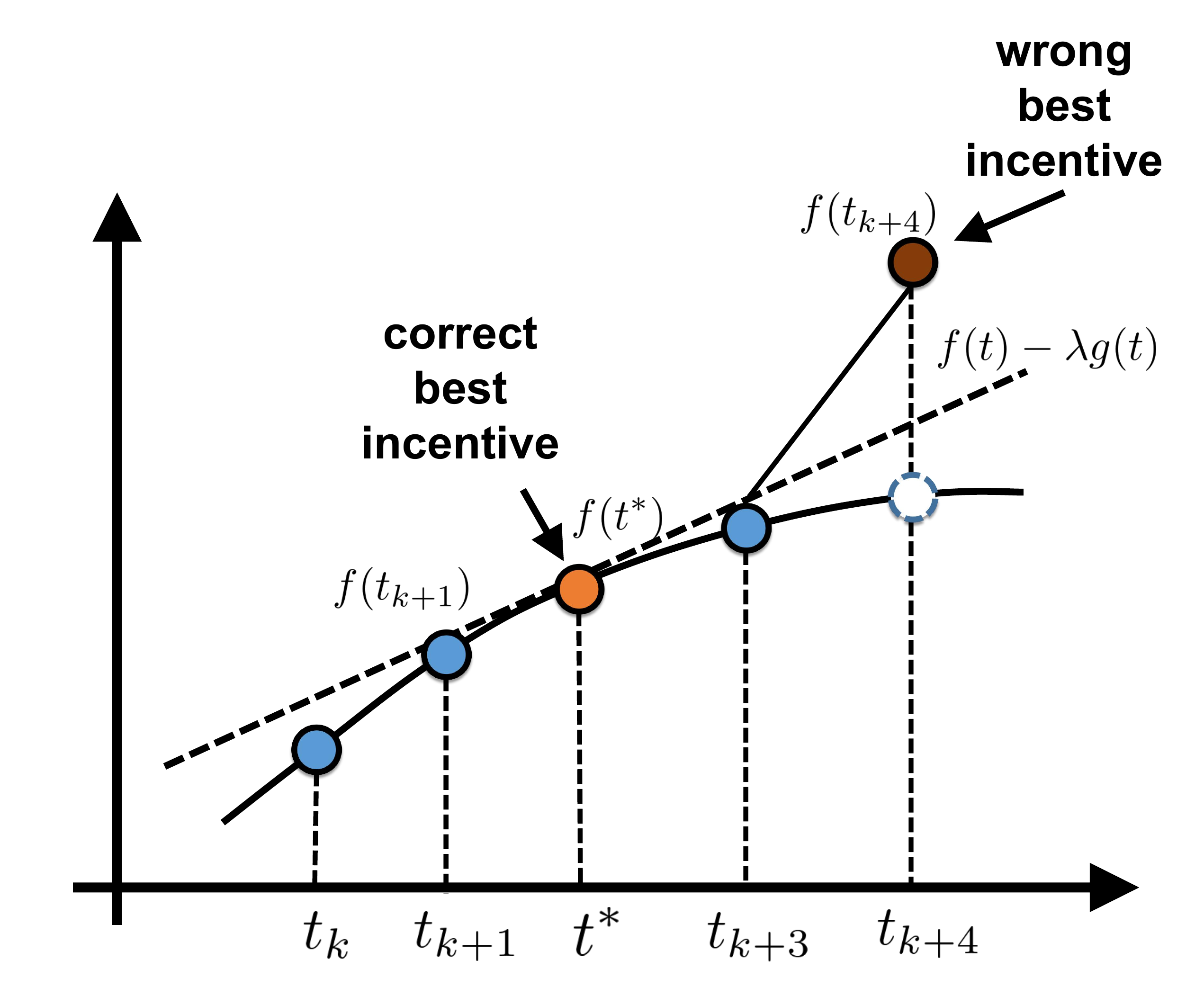}}
    \caption{An example of the common bad cases of the user response curve. $f(t)$ and $g(t)$ are a user's response function and cost function, respectively. $t^*$ is the expected best incentive level, which satisfies the $f(t^*)-\lambda g(t^*)>f(t_i)-\lambda g(t_i), \forall t_i\neq t^*$. In the two cases, we may find the wrong best incentive level $t_j$ (\textit{i.e.}, $t_{k+1}$ in (a) and $t_{k+4}$ in (b)) because of $f(t^*)-\lambda g(t^*)<f(t_j)-\lambda g(t_j), \forall t_j\neq t^*$.}
    \Description[<Figure 1. Fully described in the text.>]{<A full description of Figure 1 can be found in the third paragraph of Section 1.>}
    \label{fig:example}
\end{figure}
Our contributions can be summarized as follows:

\begin{itemize}
    \item  We propose a novel E$^3$IR model for the end-to-end optimization of the budget allocation problem.
    \item We design an uplift prediction module injected with the marketing domain knowledge, which can ensure the user response curve is monotonic and smooth. 
    \item We employ the differentiable ILP layer for the budget allocation, which can mitigate the performance gap in the two-stage methods.
    \item We conduct extensive experiments on a public dataset and a production dataset, and the results demonstrate the superiority of our method. 
\end{itemize}
\section{Related Work}
\subsection{Budget Allocation}
In recent years, several budget allocation methods have been proposed to address the issue of online marketing. Previous studies often utilize heuristic methods to determine the optimal incentive for users, considering predicted uplift~\cite{guelman2015uplift, zhao2017uplift, zhao2019uplift}. 
However, the lack of a well-defined formulation of the optimization problem may restrict the effectiveness of these methods in achieving the marketing objective.
In recent years, the most popular way of allocating budget is the two-stage methods~\cite{zhao2019unified,goldenberg2020free,ai2022lbcf,albert2022commerce}, for the first stage causal inference methods are used to predict the treatment effects. 
For the second stage, the integer programming is invoked to find the optimal allocation. 
Makhijani et al.~\cite{makhijani2019lore} introduces the marketing goal as a Min-Cost Flow network optimization problem to enhance expressiveness.
Albert et al.\cite{albert2022commerce} and Ai et al.~\cite{ai2022lbcf} employ the Multiple Choice Knapsack Problem (MCKP) to model the discrete budget allocation problem, and they propose efficient solvers based on Lagrangian duality.
Despite the effectiveness of the above methods, solutions derived from these two-stage methods may be suboptimal due to misalignment between their objectives.
There are also some works that examine the policy to solve this problem.
Du et al.~\cite{du2019improve} and Zou et al.~\cite{zou2020heterogeneous} proposed the direct learning approaches for determining the ratios between values and costs in a binary treatment setting, where treatments are first applied to users with higher scores.
Zhou et al.~\cite{zhou2023direct} proved that the proposed loss in the studies by \cite{du2019improve,zou2020heterogeneous} is unable to achieve the correct rank when it converges
The method presented in \cite{goldenberg2020free} is limited in applicability as it assumes a strict budget constraint of ROI $\geq$ 0, which lacks flexibility for numerous marketing campaigns on online platforms.
However, the proposed method relies on the assumption of diminishing marginal utility and a presumption that is frequently not strictly upheld in practical scenarios.
Our work focuses on the uplift prediction with domain knowledge, which imposes constraints on the output of user responses to be monotonic and smooth while optimizing budget allocation through an end-to-end approach.

\subsection{End-to-End Optimization}
End-to-end optimization is a crucial technique for mitigating the performance gap in the two-stage "predict + optimize" problem.
Several recent studies have investigated the integration of predictive algorithms and optimization problems. In this context, a promising approach is to utilize neural networks to differentiate the optimization layers.
Amos and Kolter~\cite{amos2017optnet} introduced a differentiable layer for Quadratic Programming (QP) optimization by differentiating the KKT conditions.
Donti et al.~\cite{donti2017task} introduced the task-based end-to-end training process for QP problems. Studies on end-to-end training of Linear Programming (LP) problems have also been explored using quadratic regularization terms~\cite{wilder2019melding}.
An interior-point approach (IntOpt)~\cite{mandi2020interior} is proposed with the homogeneous self-dual of LP problems to obtain backward gradients. 
Guler et al.~\cite{guler2022divide} propose a divide-and-conquer algorithm for solving non-convex problems in the ``predict + optimize'' framework.
Inspired by \cite{paulus2021comboptnet}, we have developed tailored predictive optimization solutions for uplift modeling in online marketing, specifically targeting the situation with budget constraints with an end-to-end optimization, which has not been extensively explored in prior works.

\section{Preliminaries}
\label{sec:pre}
In this work, we address the personalized assignment of incentives (treatments) on an online platform. 
The optimization target is to maximize the overall incremental number of customers completing a purchase. 
We can pick at most one incentive to offer each customer from a finite set of eligible incentives (see the example in Figure \ref{fig:example}). A global budget constrains the overall incremental revenue generated by the incentive.

\subsection{Cost-aware Uplift Modeling}
Let $\mathcal{D}$ be the observed dataset with $n$ samples, where each sample is represented as $(\boldsymbol{x}_i,t_i,y^c_i, y^r_i)$. 
Without loss of generality, $\boldsymbol{x}_i \in \mathcal{X}$, where $\mathcal{X} \subset \mathbb{R}^d$, is a $d$-dimensional feature vector. 
Similarly, $y^c_i$ and $y^r_i$ represent the conversion response variable and the revenue response variable, respectively. And they belong to the value spaces $\mathcal{Y}^c$ and $\mathcal{Y}^r$, which can be either binary or continuous. 
The treatment variable $t_i$ takes values in the set $\mathcal{T} = \{0, 1, \ldots, K\}$, where $K \geq 2$. For example, $\mathcal{T}$ can represent different discounts.

We formalize the problem by following the Neyman-Rubin potential outcome framework~\cite{rubin2005causal}, which enables us to express the uplift of incentive as follows. Let $y^c_i(k)$ and $y^r_i(0)$ represent the potential outcomes for user $i$ when they receive incentive $t_i = k \in \{1, \ldots, K\}$ or when they are not treated, respectively. We use $\tau^c_{i,k}$ to denote the incremental uplift caused by a specific treatment $k$. Under certain assumptions~\cite{zhang2021unified}, we can use the conditional average treatment effect (CATE) as an unbiased estimator for the uplift. CATE is defined as:
\begin{equation}
\begin{aligned}
    & \tau^c_{i,k} = \mathbb{E}\left(y^c_i(k)-y^c_i(0)\mid \boldsymbol{x}_i\right),\\
   & \tau^r_{i,k} = \mathbb{E}\left(y^r_i(k)-y^r_i(0)\mid \boldsymbol{x}_i\right).
\end{aligned}
\end{equation}
where $\tau^c_{i,k}$ represents the incremental effect on the expected conversion probability of user $i$ with assigned treatment $t=k$. $\tau^r_{i,k}$ represents the incremental effect on the expected revenue of user $i$ with assigned treatment $t=k$. Since most marketing incentives have a positive effect on the user response, thus we have $\tau^c_{i,k} \geq 0$ and $\tau^r_{i,k}\geq 0, \forall k \in \{1, \ldots K\}$. Specially, $k=0$ represents the user receives no incentive, then $\tau^c_{i,k}=0$ and $\tau^r_{i,k}=0$.

Specially, most research studies treat the cost of implementing the treatment as fixed. When considering incentives, researchers typically use the fixed coupon value as the cost of the incentive. However, this formulation fails to accurately account for the incremental profit generated by the behavioral changes induced by the incentive. For instance, while two users might exhibit the same incremental effect, one could have a higher usage level and contribute a more significant incremental profit. This aspect can only be captured by measuring the treatment's impact on cost~\cite{du2019improve}.

\subsection{Budget Allocation Problem}
\label{sec:budget}
\subsubsection{Binary Treatment Budget Allocation}
The binary treatment budget allocation problem is to assign the binary treatment to part of the individuals to maximize the overall revenue on the platform, but requires that the incremental cost does not exceed a limited budget $B$. 
Let the decision variables be the $z_i \in\{0,1\}$. Therefore, this problem can be formulated as an integer programming problem.
\begin{equation}
\begin{aligned}
&-\min \sum_i \tau^r_i z_i  \\
&\text { s.t. } \sum_i \tau^c_i z_i  \leq B \\
& z_i \in\{0,1\}, \quad i = 1, \dots, n.
\end{aligned}
\end{equation}
we take an equivalent transformation from $\max$ to $-\min$ for fitting the standard 0-1 knapsack problem,

\subsubsection{Multi-Treatment Budget Allocation}
The goal of this optimization task, which is an extension of the binary treatment budget allocation, is to select a single incentive $k^*$ for each user $i$ to maximize the sum of the selected $\tau^c_{i,k}$ values. However, the total sum of the selected $\tau^c_{i,k}$ should not exceed the budget constraint $B$.

Refer to the binary treatment budget allocation, the budget allocation in the multi-treatment scenario can be formulated as the MCKP:
\begin{equation}
\begin{aligned}
-\min \quad & \sum_{i=1}^{I}{\sum_{k=1}^{K}{\tau^r_{ik}z_{ik}}} \\
\text {s.t.} \quad & \sum_{i=1}^{I}\sum_{k=1}^{K} {\tau^c_{ik} z_{ik}} \le B \\
& \sum_{k=1}^{K}{z_{ik}} = 1, \quad i = 1, \dots, n,  \\
& z_{ik} \in \{ 0, 1 \}, \quad i = 1, \dots, n, \quad k = 1, \dots, K.
\end{aligned}\label{eq:mckp}
\end{equation}
where $z_{ik}$ is a binary assignment variable indicating whether a user $i$ is provided the $k$-th incentive or not. 
$\tau^r_{ik}$ represents the response uplift, $\tau^c_{ik}$ represents the cost uplift.

\section{Methodology}
At a high level, our E$^3$IR consists of two modules: the uplift prediction module and the differentiable allocation module. 
The first module is to predict uplifts in the user cost and response. 
Next, the second module utilizes the predicted response uplift and cost uplift to facilitate budget allocation through a differentiable optimization procedure based on an allocation objective.
The whole structure of our E$^3$IR is shown in Figure~\ref{fig:model}, and we introduce the two modules in the following.
\begin{figure*}[!t]
    \centering
    \includegraphics[width=\linewidth]{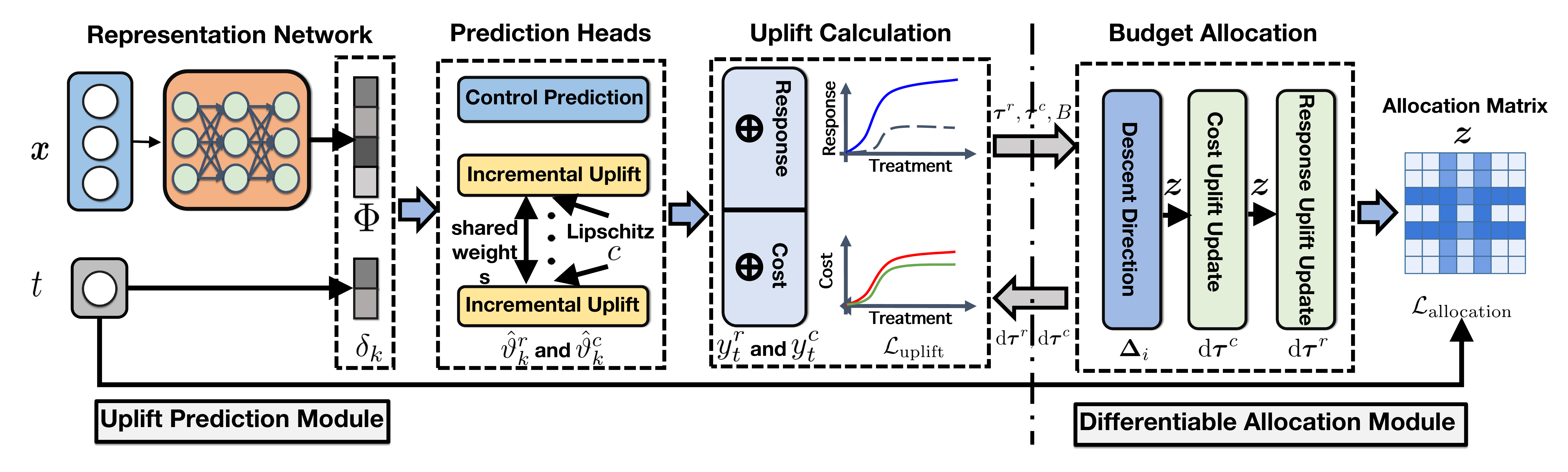}
    \caption{The overall structure of our E$^3$IR. The uplift prediction module generates the predicted uplift of user responses and corresponding costs, the differentiable allocation module generates the allocation matrix, and the two loss functions $\mathcal{L}_{\text{uplift}}$ and $\mathcal{L}_{\text{allocation}}$ are jointly optimized in Eq.~\eqref{eq:finalloss}.}
    \Description[<Figure 2. Fully described in the text.>]{<A full description of Figure 2 can be found in Section 4.>}
    \label{fig:model}
\end{figure*}

\subsection{Uplift Prediction Module}

In principle, uplift prediction for both $\tau^c_{i,k}$ and $\tau^r_{i,k}$ are achieved using uplift modeling and data from randomized controlled trials. 
To achieve the aim of this module, the previous two-stage methods use any off-the-set uplift modeling methods, such as CFRNet~\cite{shalit2017estimating} and DragonNet~\cite{shi2019adapting}, EUEN~\cite{ke2021addressing} etc.
Different from them, we inject the marketing domain knowledge (\textit{i.e.}, monotonicity, and smoothness) to design a new uplift prediction structure. 
Specifically, we formulate two constraints on the uplift prediction, \textit{i.e.}, monotonic constraint and smooth constraint.

\subsubsection{Monotonic Constraint}
Obviously, the incentive is a direct discount on the order amount in our problem setting and does not include personalized creative content. Therefore, the higher the amount, the more attractive it is for user conversion. 
However, without the unique design for the model structure, we cannot guarantee the response curve for a user $i$ is monotonic.
To ensure monotonicity, we propose a multi-treatment effect estimation model with a monotonic constraint. We implicitly model the incremental effect between two adjacent treatments and obtain the predicted uplifts through accumulation. 

As described in Section \ref{sec:pre}, we use $\left(\boldsymbol{x}_i, t_i, y_i^c, y_i^r\right)$ to represent a user, for the online platform, $x_i$ includes the statistical characteristics of the user's historical behavior, such as the number of transaction orders and the average amount of orders in the past $n$ days, and user identity descriptions, such as age and membership status, consumption preference level for each business line, etc. 
Without loss of generality, we omit the subscript of the variables.

We adapt the commonly used share bottom structure. The output of the shared bottom $\boldsymbol{L}$ is defined as $\Phi=\boldsymbol{L}(x)$, where $\Phi$ denotes the shared user representation learned from the shared bottom. It enables the efficient sharing of information across multi-treatments.
As shown in Figure \ref{fig:model}, we use $\hat{\vartheta}_i^c$ and $\hat{\vartheta}_i^r$ to represent the incremental effect of a user between two adjacent treatments. 
Thus, for the response prediction, to better distinguish the output of each head, we add the treatment embedding as part of the input rather than only separate the prediction.
Then, the input of the prediction head includes two parts: the feature representation $\Phi$ and the treatment information $\delta_i^t$. Especially, $\delta_i^t$ represents the embedding of $k$-th treatment.
Then we have:
\begin{equation}
\begin{aligned}
\hat{\vartheta}^r_i&=h^r\left(\left[\Phi ; \delta_i\right]\right), \forall k = 1,\cdots, K,\\
\hat{\vartheta}^c_i&=h^c\left(\left[\Phi ; \delta_i\right]\right), \forall k = 1,\cdots, K.
\end{aligned}
\end{equation}
where $[\cdot ; \cdot]$ denotes the concatenation of two vectors, the functions $h_i^r(\cdot)$ and $h_i^c(\cdot)$ is the incremental effect prediction head for the response and the cost, respectively.

Finally, the response prediction of a user for the $k$-th treatment can be formulated as:
\begin{equation}
\hat{y}^r_t= \begin{cases}\hat{y}^r_0, & \text { if } t=0 \\ \hat{y}^r_0+\sum_{k=1}^t \hat{\vartheta}^r_i, & \text { if } t\geq 1 .\end{cases}
\end{equation}
Following the above equation, we can also predict the cost similarly.

Then, to implement the monotonic constraint, we only need to make the incremental effect $\hat{\vartheta}^c_i$ and $\hat{\vartheta}^r_i$ greater than 0. 
It is easily achieved by squaring the raw output of the last head layer, or we can turn the output of the last head layer into an exponential form. The loss function of the proposed model is defined as follows:
\begin{equation}
\mathcal{L}_{\text{predict}}=\frac{1}{N} \sum_{(x, t, y^c,y^r) \in \mathcal{D}}\left(\mathcal{L}_c(y^c,\hat{y}^c) + \mathcal{L}_r(y^r,\hat{y}^r)\right).
\label{eq:uplift}
\end{equation}
where $N$ is the number of samples in the dataset $\mathcal{D}$, $\mathcal{L}_c$ and $\mathcal{L}_r$ are the loss functions to predict the cost and the response, which can be either cross-entropy or mean square error. $\hat{y}_c$ and $\hat{y}^r$ are the predicted cost and response for the corresponding treatment $t$.

\subsubsection{Smooth Constraint}

As illustrated in Figure~\ref{fig:example}(b), smoothness means that the model's output does not change much as the incentive level varies.
To achieve the objective, inspired by the Siamese neural network~\cite{koch2015siamese}, which is often used in contrastive learning, we share the parameter weights between different response prediction heads so as to the cost prediction heads.

Moreover, to ensure the uplift prediction can accurately reflect the influence of different treatments, we leverage Lipschitz regularization, where $h$ is Lipschitz continuous if there exists a constant $c \geq 0$ such that:
\begin{equation}
\left|h\left(v_i\right)-h\left(v_j\right)\right| \leq c\left\|v_i-v_j\right\|_2    
\end{equation}
where $v_i=\left[\Phi ; \delta_i^t\right]$ is the concatenation of the user feature embedding and the treatment embedding, $h$ can be either $h_r$ or $h_c$, and intuitively, the change of outcome is bounded by constant $c$ for smoothness. It is evident that if the neural network $f_{\Omega}$ is $c$-Lipschitz on the $v$, it is also $c$-Lipschitz on the treatment embedding $\delta_i^t$.

The Lipschitz Regularization~\cite{oberman2018lipschitz} solely depends on the weight matrices of each network layer with estimated per-layer Lipschitz upper bound $c_j$, the loss for regularizing the differences between treatment effects and outcomes can be written as:
\begin{equation}
\mathcal{L}_{\text {Lip }}=\prod_{j=1}^l \operatorname{softplus}\left(c_j \right),
\end{equation}
where $\operatorname{softplus}\left(c_j\right)=\ln \left(1+e^{c_j}\right)$ is a reparameterization strategy to avoid invalid negative estimation on Lipschitz constant $c_j$ and $l$ is the number of network layers.

Then, we can get the loss function of the uplift prediction module:
\begin{equation}
    \mathcal{L}_{\text{uplift}}=\mathcal{L}_{\text{prediction}} + \alpha \mathcal{L}_{\text{Lip}}.
\end{equation}
where $\alpha$ is the hyperparameter to control the trade-off.

\subsection{Differentiable Allocation Module}
After we have got the predicted response uplift and the cost uplift, we aim to build a differentiable incentive recommendation layer for the end-to-end optimization, which can reduce the performance gap between the uplift prediction and the budget allocation. 
Especially motivated by the ``predict + optimize'' combination optimization problem~\cite{paulus2021comboptnet}, we incorporate an ILP as a differentiable module in the model structure that inputs both constraint and objective. 

For ease of understanding, we omit the sum and the subscript in Section \ref{sec:budget} and use the vector form (bold form) of the variables. 
For the differentiability problem, we can reformulate the gradient calculation as a task of determining the descent direction.
We need to encounter the problem that the suggested gradient update $\boldsymbol{z}-\mathrm{d} \boldsymbol{z}$ to the optimal solution $\boldsymbol{z}$ is often unattainable, meaning that $\boldsymbol{z}-\mathrm{d} \boldsymbol{z}$ does not represent a feasible integer point.
\subsubsection{Descent direction} 
During the backward pass, the gradients of the subsequent layers are obtained from the ILP solver. We need to determine the change direction for the cost and response uplift. Specifically, we want the solution of the updated ILP to move in the direction opposite to the incoming gradient, also known as the direction of descent.

Given a loss function denoted by $\mathcal{L}$, let $\boldsymbol{\tau}^c$, ${B}$, $\boldsymbol{\tau}^r$, and the incoming gradient $\mathrm{d} \boldsymbol{z}=\partial \mathcal{L} / \partial \boldsymbol{z}$ at the point $\boldsymbol{z}(\boldsymbol{\tau}^c, {B}, \boldsymbol{\tau}^r)$ be provided.
To achieve end-to-end optimization, we must return gradients for $\partial \mathcal{L} / \partial \boldsymbol{\tau}^c$, $\partial \mathcal{L} / \partial {B}$, and $\partial \mathcal{L} / \partial \boldsymbol{\tau}^r$.
Our objective is to determine the directions $\mathrm{d} \boldsymbol{\tau}^c$, $\mathrm{d} B$, and $\mathrm{d} \boldsymbol{\tau}^r$ that result in the most significant decrease in distance between the updated solution $\boldsymbol{z}(\boldsymbol{\tau}^c-\mathrm{d} \boldsymbol{\tau}^c, B-\mathrm{d} B, \boldsymbol{\tau}^r-\mathrm{d} \boldsymbol{\tau}^r)$ and the target $\boldsymbol{z}-\mathrm{d} \boldsymbol{z}$.
If the mapping $\boldsymbol{z}$ is differentiable, it results in the correct gradients, such as $\partial \mathcal{L} / \partial \boldsymbol{\tau}^c=\partial \mathcal{L} / \partial \boldsymbol{z} \cdot \partial \boldsymbol{z} / \partial \boldsymbol{\tau}^c$ (similarly for $B$ and $\boldsymbol{\tau}^c$). Then we have the following proposition:
\begin{proposition}
Let $y:\mathbb{R}^{\ell} \rightarrow \mathbb{R}^n$ be a differentiable function at $x \in \mathbb{R}^{\ell}$. Let $L:\mathbb{R}^n \rightarrow \mathbb{R}$ be a differentiable function at $y = y(x) \in \mathbb{R}^n$. Denote $dy = \frac{\partial L}{\partial y}$ at $y$. Then the distance between $y(x)$ and $y - dy$ is minimized along the direction $\frac{\partial L}{\partial x}$, where $\frac{\partial L}{\partial x}$ is the derivative of $L(y(x))$ at $x$.
\end{proposition}
\begin{proof}
For any $\xi \in \mathbb{R}^{\ell}$, let $\varphi(\xi)$ represent the distance between $\boldsymbol{y}(\boldsymbol{x}-\xi)$ and the target $\boldsymbol{y}(\boldsymbol{x})-\mathrm{d} \boldsymbol{y}$. It can be formulated as $$
\varphi(\xi)=\|\boldsymbol{y}(\boldsymbol{x}-\xi)-\boldsymbol{y}(\boldsymbol{x})+\mathrm{d} \boldsymbol{y}\|.
$$

If $\mathrm{d} \boldsymbol{y}=0$, there is nothing to prove as $\boldsymbol{y}(x)=\boldsymbol{y}-\mathrm{d} \boldsymbol{y}$ and no improvements can be made. Otherwise, $\varphi$ is a positive and differentiable function in the neighborhood of zero. The Frenet derivative of $\varphi$ is given by: $$
\varphi^{\prime}(\xi)=\frac{-[\boldsymbol{y}(\boldsymbol{x}-\xi)-\boldsymbol{y}(\boldsymbol{x})+\mathrm{d} \boldsymbol{y}] \cdot \frac{\partial \boldsymbol{y}}{\partial \boldsymbol{x}}(\boldsymbol{x}-\xi)}{\|\boldsymbol{y}(\boldsymbol{x}-\xi)-\boldsymbol{y}(\boldsymbol{x})+\mathrm{d} \boldsymbol{y}\|},
$$ 
thus,
$$
\varphi^{\prime}(0)=-\frac{1}{\|\mathrm{~d} \boldsymbol{y}\|} \frac{\partial L}{\partial \boldsymbol{y}} \cdot \frac{\partial \boldsymbol{y}}{\partial \boldsymbol{x}}=-\frac{1}{\|\mathrm{~d} \boldsymbol{y}\|} \frac{\partial L}{\partial \boldsymbol{x}},
$$ 
where the last equality follows from the application of the chain rule, therefore, the steepest descent direction coincides with the derivative $\partial L / \partial \boldsymbol{x}$, considering that $\|\mathrm{d} \boldsymbol{y}\|$ is a scalar.
\end{proof}

However, every ILP solution $z\left(\tau^c-d \tau^c, B-d B, \tau^r-d \tau^r\right)$ is confined to integer points, and its ability to approach the point $z-d z$ is limited unless $d z$ itself is also an integer point. To accomplish this, we can decompose the vector $d z$ as follows:
\begin{equation}
\mathrm{d} \boldsymbol{z}=\sum_{k=1}^n \lambda_i \Delta_i.   \label{eq:decompose}
\end{equation}
where $\Delta_i \in\{-1,0,1\}^n$ are integer points and $\lambda_i \geq 0$ are scalars. The specific choice of basis $\Delta_i$ will be discussed separately. For now, it is sufficient to understand that each point $z_i^{\prime}=z-\Delta_i$ is an integer point neighboring $y$ that indicates the direction of $-d z$. Then, we address separate problems by replacing $d z$ with the integer updates $\Delta_i$.

To maintain the linearity of the standard gradient mapping, our objective is to combine the gradients from the subproblems to create the final gradient linearly. It is important to note that in the budget allocation problem, $B$ is a constant. As a result, $\mathrm{d} B=0$, and we do not update $B$ in the subsequent description.

\subsubsection{Cost uplift update} 
To obtain a significant update for an achievable change $\Delta_i$, we compute the gradient of a piecewise affine local mismatch function $P_{\boldsymbol{z}_i^{\prime}}$. 
The definition of $P_{\boldsymbol{z}_i^{\prime}}$ is derived from a geometric understanding of the underlying structure. In doing so, we depend on the Euclidean distance between a point and a hyperplane. 
Indeed, for any point $\boldsymbol{z}$ and a given hyperplane, parametrized by vector $\boldsymbol{\tau}^c$ and scalar $B$ as $\boldsymbol{z} \mapsto \boldsymbol{\tau}^c \cdot \boldsymbol{z}-B$, we have:
\begin{equation}
\operatorname{dist}(\boldsymbol{\tau}^c, B ; \boldsymbol{z})=|\boldsymbol{\tau}^c \cdot \boldsymbol{z}-B| /\|\boldsymbol{\tau}^c\| .  
\end{equation}
To update the constraints, we define that if $\boldsymbol{\tau}^c \boldsymbol{z}_i^{\prime} \leq \boldsymbol{b}$, then $\boldsymbol{z}_i^{\prime}$ is feasible. Then $P_{\boldsymbol{z}_i^{\prime}}(\boldsymbol{\tau}^c, B)$ is formulated as:
\begin{equation}
P_{\boldsymbol{z}_i^{\prime}}(\boldsymbol{\tau}^c, B)=\left\{\begin{array}{c}\min _j \operatorname{dist}\left(\boldsymbol{\tau}^c_j, B ; \boldsymbol{z}\right),  \text { if } \boldsymbol{z}_i^{\prime} \text { is feasible and } \boldsymbol{z}_i^{\prime} \neq \boldsymbol{z} \\ \sum_j \llbracket \boldsymbol{\tau}^c_j \cdot \boldsymbol{z}_i^{\prime}>B \rrbracket \operatorname{dist}\left(\boldsymbol{\tau}^c_j, B ; \boldsymbol{z}_i^{\prime}\right), \text { if }\boldsymbol{z}_i^{\prime} \text{ is infeasible}\\
0, \text { if }\boldsymbol{z}_i^{\prime}=\boldsymbol{z}\text { or }\boldsymbol{z}_i^{\prime} \notin \mathcal{Z}.
\end{array}\right.    
\end{equation}
$\mathcal{Z}$ is the value space of $\boldsymbol{z}$. Imposing linearity and using Eq.~\eqref{eq:decompose}, we define the gradient $\mathrm{d} \boldsymbol{\tau}^c$ as:
\begin{equation}
\mathrm{d} \boldsymbol{\tau}^c=\sum_{k=1}^n \lambda_i \frac{\partial P_{\boldsymbol{z}_i^{\prime}}}{\partial \boldsymbol{\tau}^c}(\boldsymbol{\tau}^c, B).\label{eq:gradient}
\end{equation}


Note that the mapping $\mathrm{d} \boldsymbol{z} \mapsto \mathrm{d} \boldsymbol{\tau}^c$ is homogeneous. 
It is because the whole situation is rescaled to one case (choice of basis) where the gradient is computed and then rescaled back (scalars $\lambda_i$). The most natural scale agrees when $\boldsymbol{z}_i^{\prime}$ are the closest integer neighbors. 
This is due to the scale, which means that the entire situation is rescaled to one case (choice of basis), where the gradient is computed and subsequently rescaled. The most natural scale occurs when $\boldsymbol{z}_i^{\prime}$ are the nearest integer neighbors. This ensures the situation does not collapse into a trivial solution (zero gradients) and prevents interference with distant values of $\boldsymbol{z}$. The selection of this basis serves as a hyperparameter. In our case, we construct a valid basis explicitly and do not require optimization of any additional hyperparameters.

\subsubsection{Response uplift update} 
Ignoring the distinction between feasible and infeasible $\boldsymbol{z}_i^{\prime}$, the problem of updating the cost has been addressed in several prior studies. We employ a straightforward approach that defines the mismatch function in a way that the resulting update prioritizes $\boldsymbol{z}_i^{\prime}$ over $\boldsymbol{z}$ in the updated optimization problem:
\begin{equation}
P_{\boldsymbol{z}_i^{\prime}}(\boldsymbol{\tau}^r)= \begin{cases}\boldsymbol{\tau}^r \cdot\left(\boldsymbol{z}_i^{\prime}-\boldsymbol{z}\right) & \text { if } \boldsymbol{z}_i^{\prime} \text { is feasible}, \\ 0 & \text { if } \boldsymbol{z}_i^{\prime} \text { is infeasible or } \boldsymbol{z}_i^{\prime} \notin \mathcal{Z} .\end{cases}
\end{equation}
The gradient $\mathrm{d} \boldsymbol{\tau}^r$ is then composed similarly as in Eq.~\eqref{eq:gradient}.
\subsubsection{The choice of the basis} 
Let $i_1, \ldots, i_n$ be the indices of the coordinates in the absolute values of $\mathrm{d} \boldsymbol{z}$ in decreasing order, \textit{i.e.},
\begin{equation}
\left|\mathrm{d} \boldsymbol{z}_{i_1}\right| \geq\left|\mathrm{d} \boldsymbol{z}_{i_2}\right| \geq \cdots \geq\left|\mathrm{d} \boldsymbol{z}_{i_n}\right|    ,
\end{equation}
and set
\begin{equation}
\Delta_i=\sum_{j=1}^k \operatorname{sign}\left(\mathrm{d} \boldsymbol{z}_{i_j}\right) \boldsymbol{u}_{i_j}, \label{eq:delta}
\end{equation}
where $e_i$ represents the $k$-th canonical vector. Consequently, $\Delta_i$ represents the (signed) indicator vector of the initial $i$ dominant directions.

Let $\ell$ be the largest index such that $|\mathrm{d} \boldsymbol{z}_{\ell}|>0$. Consequently, the initial $\ell$ vectors $\Delta_i$ are linearly independent and serve as a basis for their respective subspace. Thus, there are scalars $\lambda_i$ that satisfy decomposition in Eq.~\eqref{eq:decompose}.

\begin{proposition}
If $\lambda_j=\left|\mathrm{d} \boldsymbol{z}_{i_j}\right|-\left|\mathrm{d} \boldsymbol{z}_{i_{j+1}}\right|$ for $j=$ $1, \ldots, n-1$ and $\lambda_n=\left|\mathrm{d} \boldsymbol{z}_{i_n}\right|$, then Eq.~\eqref{eq:decompose} holds with $\Delta_i$'s as in Eq.~\eqref{eq:delta}.  
\end{proposition}
\begin{proof}
We prove that
\begin{equation}
\sum_{j=\ell}^n w_j \boldsymbol{u}_{i_j}=\sum_{j=\ell}^n \lambda_j \Delta_j-\left|w_{\ell}\right| \sum_{j=1}^{\ell-1} \operatorname{sign}\left(w_j\right) \boldsymbol{u}_{i_j}, \label{eq:proof1}
\end{equation}
For each $\ell = 1, \ldots, n$, with the abbreviation $w_j = \mathrm{d} \boldsymbol{z}_{i_j}$, the desired equality of Eq.~\eqref{eq:decompose} can be derived from the special case where $\ell = 1$ in Eq.~\eqref{eq:proof1}. 

Next, we proceed with the proof of induction. To start, we will demonstrate the validity of Eq.~\eqref{eq:proof1} for $\ell = n$. By referring to the definition of $\Delta_n$ in Eq.~\eqref{eq:delta}, we can observe:
$$
\begin{aligned}
\lambda_n & \Delta_n-\left|w_n\right| \sum_{j=1}^{n-1} \operatorname{sign}\left(w_j\right) \boldsymbol{u}_{i_j} \\
& =\left|w_n\right| \sum_{j=1}^n \operatorname{sign}\left(w_j\right) \boldsymbol{u}_{i_j}-\left|w_n\right| \sum_{j=1}^{n-1} \operatorname{sign}\left(w_j\right) \boldsymbol{u}_{i_j} \\
& =w_n \boldsymbol{u}_{i_n}
\end{aligned}
$$

Assuming that Eq.~\eqref{eq:proof1} holds for $\ell+1 \geq 2$, then we demonstrate that it also holds for $\ell$:
$$
\begin{aligned}
\sum_{j=\ell}^n & \lambda_j \Delta_j-\left|w_{\ell}\right| \sum_{j=1}^{\ell-1} \operatorname{sign}\left(w_j\right) \boldsymbol{u}_{i_j} \\
= & \sum_{j=\ell+1}^n \lambda_j \Delta_j-\left|w_{\ell+1}\right| \sum_{j=1}^{\ell} \operatorname{sign}\left(w_j\right) \boldsymbol{u}_{i_j}+\lambda_{\ell} \Delta_{\ell} \\
& +\left|w_{\ell+1}\right| \sum_{j=1}^{\ell} \operatorname{sign}\left(w_j\right) \boldsymbol{u}_{i_j}-\left|w_{\ell}\right| \sum_{j=1}^{\ell-1} \operatorname{sign}\left(w_j\right) \boldsymbol{u}_{i_j} \\
= & \sum_{j=\ell+1}^n w_j \boldsymbol{u}_{i_j}+\left(\left|w_{\ell}\right|-\left|w_{\ell+1}\right|\right) \sum_{j=1}^{\ell} \operatorname{sign}\left(w_j\right) \boldsymbol{u}_{i_j} \\
& +\left|w_{\ell+1}\right| \sum_{j=1}^{\ell} \operatorname{sign}\left(w_j\right) \boldsymbol{u}_{i_j}-\left|w_{\ell}\right| \sum_{j=1}^{\ell-1} \operatorname{sign}\left(w_j\right) \boldsymbol{u}_{i_j} \\
= & \sum_{j=\ell+1}^n w_j \boldsymbol{u}_{i_j}+\operatorname{sign}\left(w_{\ell}\right)\left|w_{\ell}\right| \boldsymbol{u}_{i_{\ell}}=\sum_{j=\ell}^n w_j \boldsymbol{u}_{i_j},
\end{aligned}
$$
where we use the definitions of $\Delta_{\ell}$ and $\lambda_{\ell}$.
\end{proof}

Moreover, for the binary treatment allocation problem, the value space $\mathcal{Z}$ of $\boldsymbol{z}$ is $\{0, 1\}^n$. 
For multi-treatment, the optimal $z_{ik^*}$ should satisfy:
\begin{equation}
z_{ik^*}=\mathbbm{1}\left(k^*=\underset{k}{\arg \max }~\tau_{i k}^r-\lambda_{ik} \tau_{i k}^c \right), 
\end{equation}
where $\mathbbm{1}$ is the 0/1 indicator function. Corresponding to Eq.~\eqref{eq:mckp}, we change the decomposition results in as a vector $\boldsymbol{\Delta}_{i}$ for the multi-choice problem, where $\boldsymbol{\Delta}_{i}=[\Delta_{i1}, \Delta_{i2}, \cdots, \Delta_{iK}]$ and $\Delta_{ik} \in\{-1,0,1\} \cap \sum_k \Delta_{ik} \in\{-1,0,1\}$. And the value space $\mathcal{Z}$ of $\boldsymbol{z_i}=[z_{i1}, z_{i2}, \cdots, z_{iK}]$ is $z_{ik} \in\{0,1\} \cap \sum_k z_{ik}=1$.



Moreover, inspired by the Expected Outcome Metric in \cite{zhao2017uplift}, we design the loss function for the differentiable budget allocation module as:
\begin{equation}
    \mathcal{L}_{\text{allocation}} = \frac{1}{N} \sum_{(x, t, y^c_i,y^r_i) \in \mathcal{D}} \mathcal{L}_{d}(t, z_i),
\end{equation}
where $\mathcal{L}_{d}(\cdot, \cdot)$ is the cross entropy loss, $\boldsymbol{z}_i$ is the predicted incentive recommendation for each user. Combing with Eq.~\eqref{eq:uplift}, the total loss for our E$^{3}$IR is:
\begin{equation}
    \mathcal{L}_{\text{E$^{3}$IR}} = \mathcal{L}_{\text{predict}}+ \beta\mathcal{L}_{\text{allocation}}.\label{eq:finalloss}
\end{equation}
where $\beta$ is the hyperparameter to control the trade-off.

    
    

\subsection{Discussion}
This section discusses the distinctions between our E$^3$IR and the online MCKP and offline MCKP problems~\cite{albert2022commerce,zhou2023direct}. In the online MCKP, the goal of optimization is to recommend a single incentive at each decision point, considering both the value and weight quantities. This online approach allows for dynamic adaptation of strategies. On the other hand, the offline MCKP employs the same underlying selection algorithm as the online MCKP but with items provided in advance. This allows the efficiency angle function to be fitted once before making recommendations. Subsequently, the algorithm determines which incentive to offer each customer without updating the efficiency angle function. 
In online marketing, the final objective of uplift prediction is to recommend incentives to customers. Therefore, optimizing the uplift model without incentive recommendation information may result in sub-optimal results due to the objective gap. Consequently, our E$^3$IR aims to provide incentive recommendation information for uplift prediction under budget constraints, aiming to mitigate the optimality gap between two-stage methods involving online or offline MCKP.

\section{Experiments}
In this section, we conduct experiments to answer the following research questions:
\begin{itemize}
    \item \textbf{RQ1:} How is the performance of our E$^3$IR compared with other baselines on both binary treatment and multi-treatment datasets?
    \item \textbf{RQ2:} How is the role that each part of our model plays?
    \item \textbf{RQ3:} How is the performance of our E$^3$IR  compared with other baselines when the budget changes on the multi-treatment datasets?
\end{itemize}

\subsection{Datasets}
$\bullet$ \textbf{Hillstrom}~\cite{diemert2018large}: This dataset is derived from an email merchandising campaign that involves 64,000 consumers who last purchased within twelve months. The features include recency, history\_segment, history, mens, womens, zip\_code, newbie, channel. The dataset contains three variables describing consumer activity in the following two weeks after email campaign delivery: visit, conversion, and spend. We select the spend as the response, the visit as the cost, and the segment as the treatment (3 types). Following \cite{betlei2021uplift}, we also construct two binary treatment datasets \textit{Hillstrom-Men} and \textit{Hillstrom-Women}, which combine the mens\_email treatment group and no\_email control group, the womens\_email treatment group and no\_email control group, respectively.

\noindent
$\bullet$ \textbf{Production}: This dataset was obtained from one of China's biggest short video platforms. For such platforms, video sharpening is known as a valuable source for studying user experience indicators. Different degrees of video clarity can significantly influence user experiences, potentially affecting users' playback time and bringing different bandwidth costs. To investigate this, we conducted random experiments over two weeks, assigning three levels of video sharpening ($T=1, 2, 3$) as treatment groups, while regular videos ($T=0$) served as the control group. We quantified the impact of different treatments by tracking users' short video playback time and bandwidth costs during this period. The resulting dataset comprises over 8 million users, with 108 features capturing user-related characteristics. We present the dataset collection visualization in Figure~\ref{fig:product}.
Moreover, the data statistics are shown in the appendix.

\begin{figure}[htbp]
    \centering
    \includegraphics[width=\linewidth]{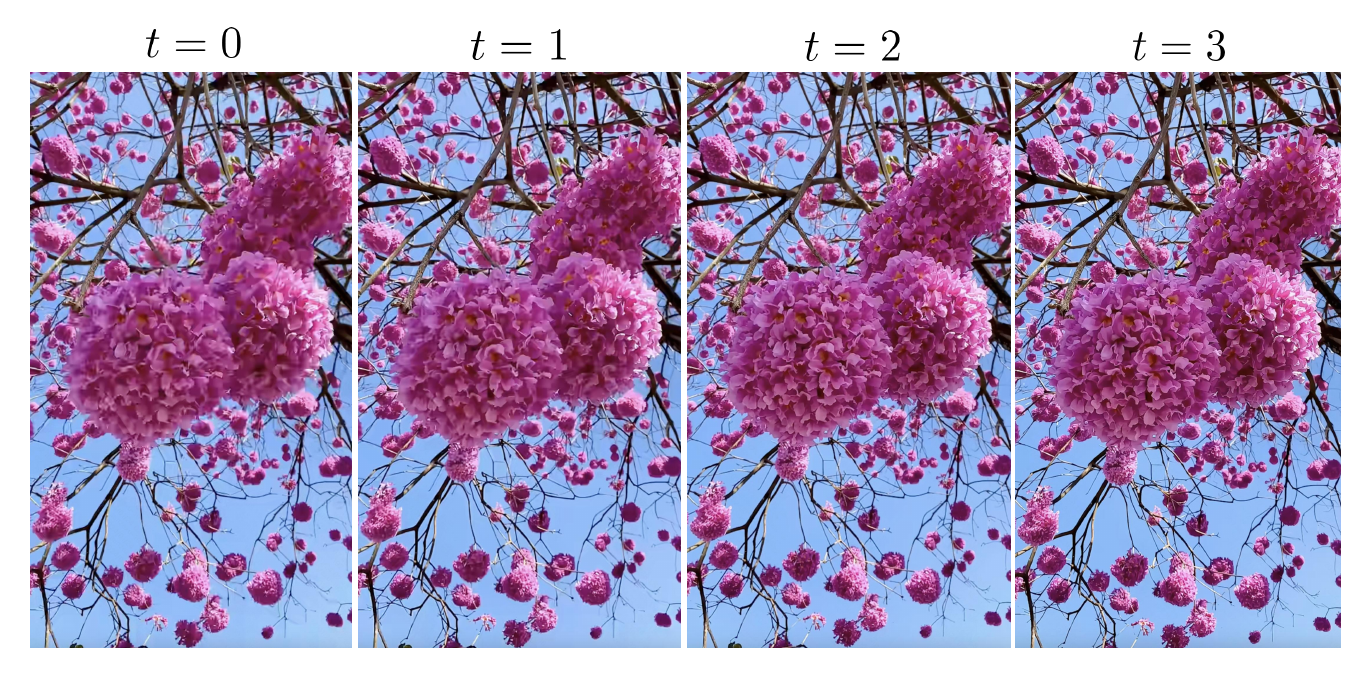}
    \caption{The visualization of the Production dataset collection subject to different treatments. As $t$ increases, the clarity of the video correspondingly enhances.}
    \Description[<Figure 3. Fully described in the text.>]{<A full description of Figure 3 can be found in Section 5.1.>}
    \label{fig:product}
\end{figure}

\subsection{Baselines and Metrics}
\subsubsection{Baselines}
We compare our E$^3$IR with the following baselines:
\begin{itemize}
    \item Cost-unaware binary incentive recommendation problem: \textbf{S-Learner}~\cite{kunzel2019metalearners}, \textbf{X-Learner}~\cite{kunzel2019metalearners}, \textbf{Causal Forest}~\cite{davis2017using}, \textbf{CFRNet}~\cite{shalit2017estimating},  \textbf{DragonNet}~\cite{shi2019adapting}, \textbf{EUEN}~\cite{ke2021addressing}.
    \item Cost-aware binary incentive recommendation problem: \textbf{TPM-SL}~\cite{zhao2019unified}, \textbf{Direct Rank}~\cite{du2019improve}, \textbf{DRP}~\cite{zhou2023direct}.
    \item Cost-aware multi-incentive recommendation problem: \textbf{Multi-TPM-SL}~\cite{ai2022lbcf}, \textbf{DPM}~\cite{zhou2023direct}.
\end{itemize}

\subsubsection{Metrics}
We evaluate our model with the following metrics: \textsc{AUUC (Area under Uplift Curve)}~\cite{rzepakowski2010decision}, \textsc{QINI (Qini Coefficient)}~\cite{mouloud2020adapting}, \textsc{KENDALL (Kendall's Rank Correlation)}~\cite{mouloud2020adapting}, \textsc{AUCC (Area under Cost Curve)}~\cite{du2019improve}, \textsc{MT-AUCC (Muti-Treatment AUCC)}~\cite{zhou2023direct}, \textsc{EOM (Expected Outcome Metric)}~\cite{zhao2017uplift}.

Due to the space limitation, we present the details of the baselines and metrics in the Appendix.

\subsection{Implementation Details}
We implement all baselines and our E$^3$IR based on Pytorch 1.10 and Jax 0.4.23, with Adam as the optimizer and a maximum iteration count of 30. We use the QINI as a reference to search for the best hyper-parameters for all baselines and our model. We also adopt an early stopping mechanism with a patience of 5 to avoid over-fitting the training set. Furthermore, we utilize the hyper-parameter search library Optuna~\cite{akiba2019optuna} to accelerate the tuning process. All experiments are implemented on NVIDIA A40 and Intel(R) Xeon(R) 5318Y Gold CPU @ 2.10GHz.

\subsection{Overall Performance (RQ1)}

For the uplift prediction and budget allocation on the binary treatment datasets (\textit{i.e.}, Hillstorm-Men, Hillstorm-Women), we present the results in Table~\ref{tab:overall1}. 
From the results, we have the following observations:
1) Meta-learners (\textit{i.e.}, the S-Learner and X-Learner) appear to perform competitively with more advanced deep learning approaches, particularly in KENDALL, which is a measure of the model's ability to rank users by the predicted uplift.
Causal Forest have a better performance than CFRNet$_{wass}$ across all the metrics, which shows the effectiveness of the ensemble structure and the splitting criterion.
2) DragonNet performs best among all the cost-unaware baselines, this may because of the design of the target regularization, which can reduce the predictive error of uplift.
In total, the representation learning-based methods have a higher KENDALL than other baselines, which suggests that representation learning may have potential advantages in capturing ranking consistency on the uplift prediction. 
3) Our E$^3$IR exhibits superior performance on both datasets across all the three metrics, this shows that the monotonic and smoothness constraints of the user response curve in online marketing can improve the performance of the ranking metrics in uplift modeling. 
Moreover, the differentiable allocation module can help train the uplift prediction module to be more effective from the decision making perspective.

From the results of cost-aware binary treatment baselines, we set the budgets as 400 in both two datasets, and we have the following observations: 
1) TPM-SL performs similar to the S-Learner on the basic uplift metrics (\textit{i.e.}, AUUC, QINI and KENDALL), because the TPM-SL uses the same model structure as the S-Learner for uplift prediction. 
However, this method shows the worst performance on AUCC among all cost-aware binary treatment methods; this may be because using the S-Learner as the estimator of response and cost predictors will introduce a bigger prediction error than other methods, which leads to a bad AUCC.
2) Direct Rank and DRP performs better than most cost-unaware baselines, this is because that the two methods design the learning objective related to the cost, with the cost information, the model can get more accurate rank of ROI, which is benefit for the uplift prediction. 
DRP performs better than Direct Rank on AUCC, which leverages a factor model for mitigating the performance gap between the prediction and optimization. This is helpful for getting better budget allocation metrics. 
3) Our E$^3$IR shows superior performance across diverse metrics and datasets, especially on QINI. This is due to that we use the QINI as the objective to tune the hyperparameters of all the baselines and our E$^3$IR.
Compared to the cost-aware baselines, the superior performance indicates that the differentiable budget allocation module can help our E$^3$IR get better results on the ROI ranking. 
The results also show that the end-to-end training of our E$^3$IR can reduce the performance gap between the prediction and optimization.

Considering the multi-treatment scenarios, we evaluate the performance of the budget allocation on the Hillstorm and Production datasets with the budget as 500 and 50 thousand, respectively.  
We present the results in Table~\ref{tab:overall2}, from the results, we have the following observations:
1) DRM achieves the suboptimal performance on most results of the two datasets, this may because the design of decision factor module. 
Compared with solving the ILP problem by the bisection methods, the decision factor can transport the optimization information into the uplift prediction, and this can partly remove the performance gap in the two-stage method Multi-TPM-SL.  
2) Our E$^3$IR achieves superior performance across the MT-AUCC and EOM in the two datasets. Expected the public dataset Hillstorm, the results of the Production dataset further verify the effectiveness of our designed structure. 
The consistency of E$^3$IR’s performance accentuates its adaptability and potential for generalization across diverse experimental conditions.

\begin{table*}[htbp]
    \centering
    \caption{Overall comparison between our models and the baselines on Hillstorm Men and Hillstorm Women datasets. We report the results over five random seeds. The best and second best results are \textbf{bold} and \underline{underlined}, respectively.} 
\resizebox{\linewidth}{!}{
    \begin{tabular}{c|ccccccccc}\toprule
      \multirow{2}{*}{\textbf{Methods}}   &\multicolumn{4}{c}{Hillstorm Men} & &\multicolumn{4}{c}{Hillstorm Women} \\\cline{2-5} \cline{7-10}
         & AUUC & QINI & KENDALL & AUCC& & AUUC & QINI & KENDALL& AUCC\\\midrule
S-Learner &0.4729 $\pm$ 0.0377 & 0.0206 $\pm$ 0.0110  &0.5271 $\pm$ 0.0103 & --- & &0.4804 $\pm$ 0.0253 &0.0271 $\pm$ 0.0202 & 0.4809 $\pm$ 0.0197 & --- \\
X-Learner &0.4557 $\pm$ 0.0245& 0.0473 $\pm$ 0.0227 &0.5324 $\pm$ 0.0402 & --- & & 0.4513 $\pm$ 0.0056 &0.0421 $\pm$ 0.0332 & 0.5723 $\pm$ 0.0288&  --- \\
Causal Forest &0.5023 $\pm$ 0.0219 & 0.0487 $\pm$ 0.0165  &0.5689 $\pm$ 0.0402 & --- & &0.5122 $\pm$ 0.0104 & 0.0523 $\pm$ 0.0211 & 0.5504 $\pm$ 0.0301 & --- \\
CFRNet$_{mmd}$ &0.5228 $\pm$ 0.0306  &0.0496 $\pm$ 0.0232 &0.5451 $\pm$ 0.0243& --- & &0.5189 $\pm$ 0.0349 &0.0519 $\pm$ 0.0251& 0.5244 $\pm$ 0.0342 & --- \\
CFRNet$_{wass}$ &0.4998 $\pm$ 0.0442  & 0.0464 $\pm$ 0.0266 &0.5199 $\pm$ 0.0303& --- & &0.5097 $\pm$ 0.0369 &0.0477 $\pm$ 0.0198& 0.5308 $\pm$ 0.0314& --- \\
DragonNet &0.5640 $\pm$ 0.0392  & 0.0670 $\pm$ 0.0298 &0.6804 $\pm$ 0.0381 & --- & &0.5921 $\pm$ 0.0289 &0.0623 $\pm$ 0.0277 &0.6597 $\pm$ 0.0392 & ---  \\
EUEN &0.5723 $\pm$ 0.0122&0.0603 $\pm$ 0.0265 &0.5408 $\pm$ 0.0385 & --- & &0.5414 $\pm$ 0.0151 &0.0654 $\pm$ 0.0288& 0.5202 $\pm$ 0.0401 & ---  \\
\hline
TPM-SL & 0.4665 $\pm$ 0.0188& 0.0315 $\pm$ 0.0249 &0.5902 $\pm$ 0.0351 & 0.0476 $\pm$ 0.0056 & &0.4725 $\pm$ 0.0211 &0.0310 $\pm$ 0.0256& 0.5108 $\pm$ 0.0399 & 0.0502 $\pm$ 0.0064  \\
Direct Rank &\underline{0.5810} $\pm$ 0.0198&\underline{0.0683} $\pm$ 0.0295 &0.6406 $\pm$ 0.0299 & 0.0498 $\pm$ 0.0047 & &0.6022 $\pm$ 0.0251 &\underline{0.0671} $\pm$ 0.0265& 0.6441 $\pm$ 0.0297& 0.0519 $\pm$ 0.0048  \\
DRP &0.5783 $\pm$ 0.0252&0.0667 $\pm$ 0.0210 & \underline{0.6811} $\pm$ 0.0275 & \textbf{0.0545} $\pm$ 0.0044 & &\underline{0.6211} $\pm$ 0.0298 &0.0669 $\pm$ 0.0214& \textbf{0.6802} $\pm$ 0.0357& \underline{0.0531} $\pm$ 0.0043   \\
\hline
\rowcolor{mygray} \textbf{E$^3$IR} & \textbf{0.5928} $\pm$ 0.0193& \textbf{0.0717} $\pm$ 0.0196&\textbf{0.7033} $\pm$ 0.0353 & \underline{0.0502} $\pm$ 0.0069 &  & \textbf{0.6466} $\pm$ 0.0267 &\textbf{0.0760} $\pm$ 0.0187&\underline{0.67f54} $\pm$ 0.0322 & \textbf{0.0587} $\pm$ 0.0027
\\\bottomrule
    \end{tabular}}
    \label{tab:overall1}
\end{table*}

\begin{table*}[htbp]
    \centering
    \caption{Overall comparison between our models and the baselines on Hillstorm and Production datasets. The EOM is represented by employing the min-max normalized responses. Moreover, the EOM of the Production dataset is scaled down by a factor of $e^4$. We report the results over five random seeds. The best and second best results are \textbf{bold} and \underline{underlined}, respectively.} 
\begin{tabular}{c|ccccccc}\toprule
      \multirow{2}{*}{\textbf{Methods}}   &\multicolumn{2}{c}{Hillstorm} & &\multicolumn{2}{c}{Production} \\\cline{2-3} \cline{5-6}
         & MT-AUUC & EOM & & MT-AUUC & EOM \\\midrule
Multi-TPM-SL &0.0645 $\pm$ 0.0078 & 21.9767 $\pm$ 0.1088  & & \underline{0.3907} $\pm$ 0.0196 &27.6576 $\pm$ 0.1079   \\
DRM &\underline{0.0726} $\pm$ 0.0021& \underline{26.5744} $\pm$ 0.1102  & & 0.3601 $\pm$ 0.0258 & \underline{36.6544} $\pm$ 0.1099  \\\hline
\rowcolor{mygray} \textbf{E$^3$IR} &\textbf{0.0803} $\pm$ 0.0044 & \textbf{29.9221} $\pm$ 0.1698  & &\textbf{0.4639} $\pm$ 0.0311 
& \textbf{44.5087} $\pm$ 0.1877\\\bottomrule
    \end{tabular}
    \label{tab:overall2}
\end{table*}

\subsection{Ablation Study (RQ2)}
To analyze the role played by each proposed module, we construct the ablation study on both the binary and multi-treatment datasets. 
Specifically, we sequentially remove each module of our E$^3$IR, \textit{i.e}. the monotonic constraint in the uplift prediction module (E$^3$IR w/o MC), the smooth constraint in the uplift prediction module (E$^3$IR w/o SC), the differentiable allocation module (E$^3$IR w/o DA). 
We present the results in Figure~\ref{fig:ablation}. 
From the results, we can see that removing any module will cause performance degradation across all the datasets and metrics. This verifies the validity of each module design in our E$^3$IR. That is, the monotonic constraint and the smooth constraint can help the model get a better ranking performance of the predicted uplift by injecting the marketing domain knowledge into the structure design. 
The differentiable allocation module can build the bridge between the prediction and optimization, which can bring a big performance improvement compared with other parts in our E$^3$IR.

\begin{figure}[htbp]
    \centering
    \subfigure[Binary treatment]{\includegraphics[width=0.48\linewidth]{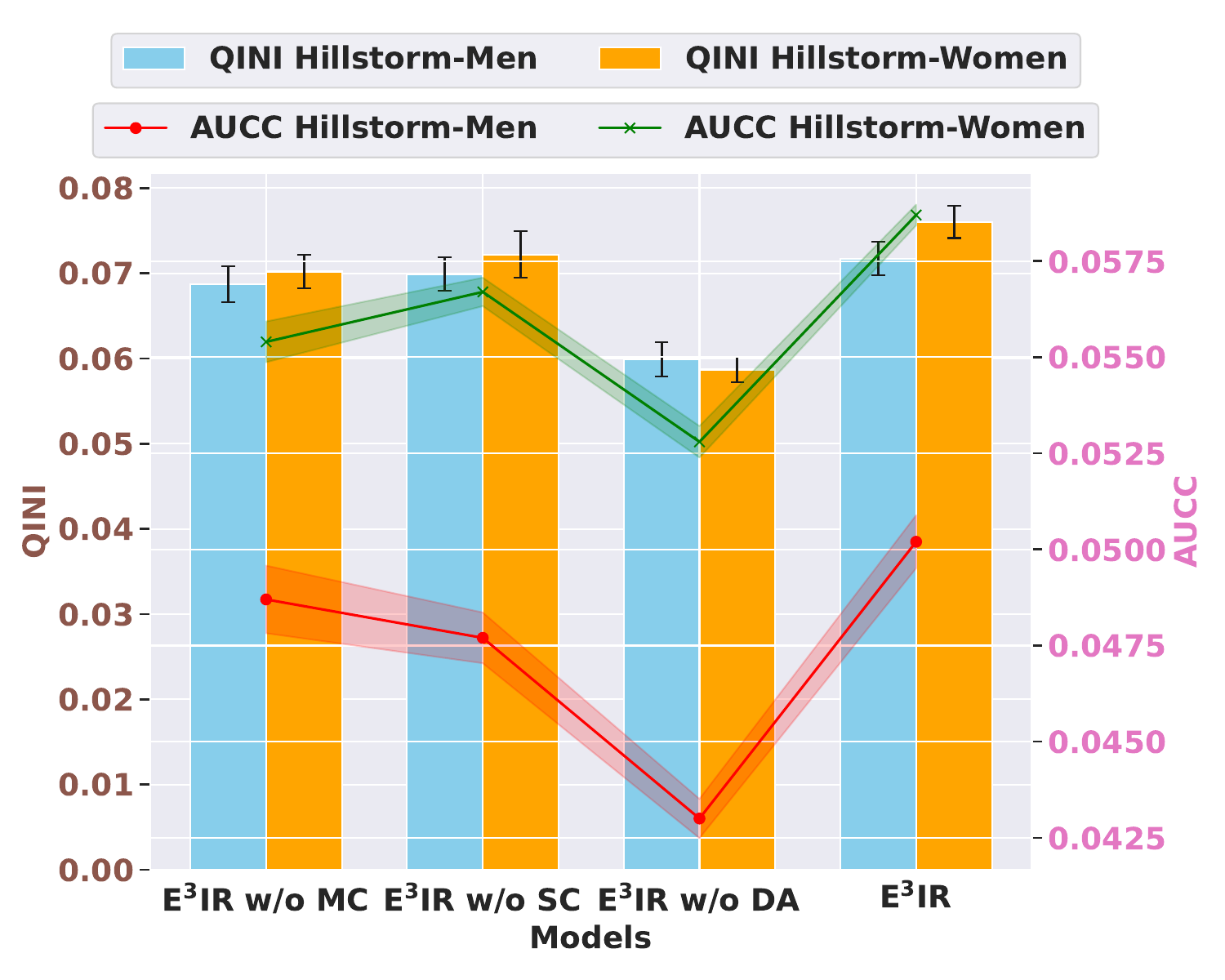}}~
    \subfigure[Multi-treatment]{\includegraphics[width=0.48\linewidth]{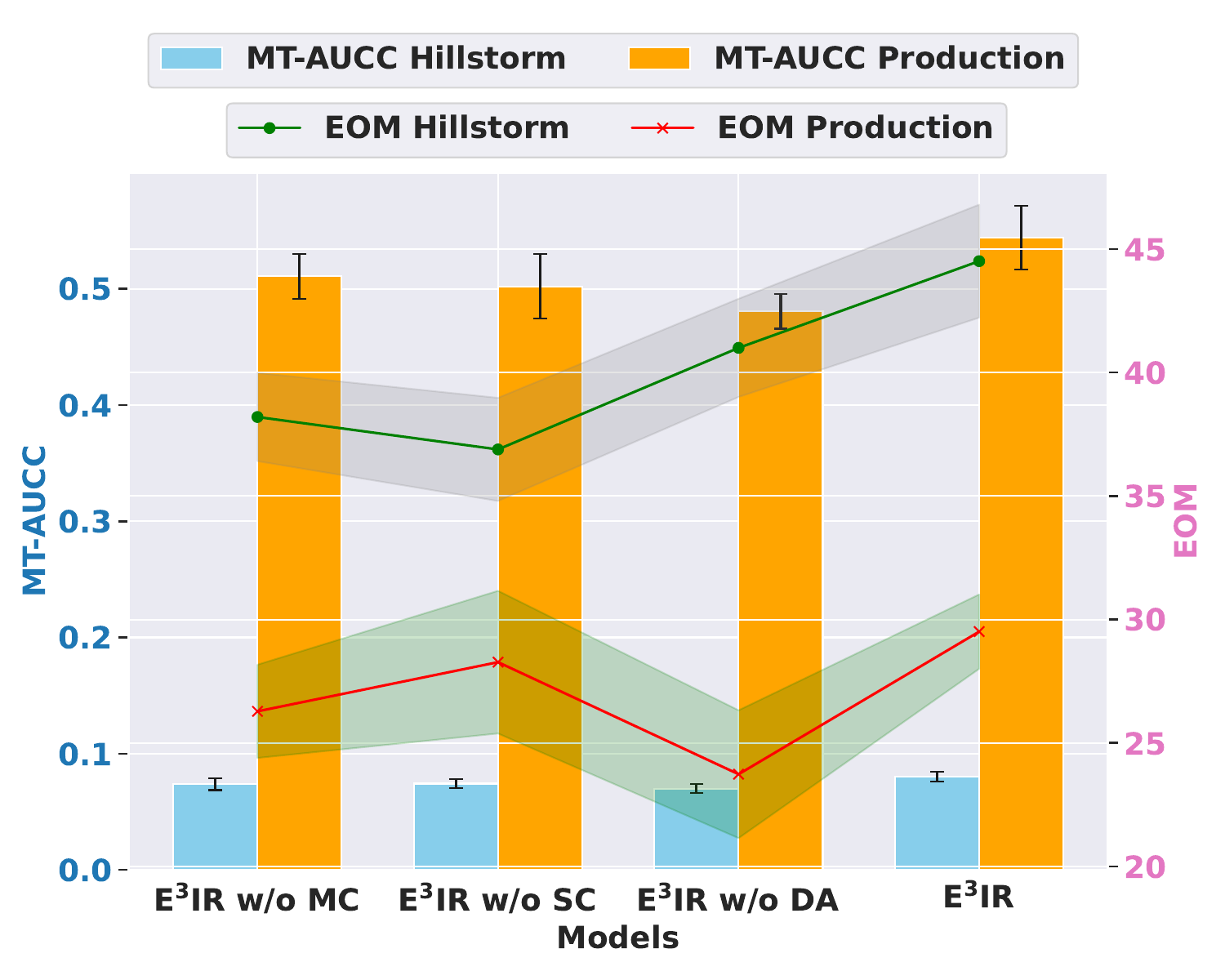}}
    \caption{Ablation study of our E$^3$IR on all the binary treatment and multi-treatment datasets.}
    \Description[<Figure 4. Fully described in the text.>]{<A full description of Figure 4 can be found in Section 5.5.>}
    \label{fig:ablation}
\end{figure}

\subsection{Analysis of Budget Influence (RQ3)}
After we predict the uplifts $\boldsymbol{\tau}^r$ and $\boldsymbol{\tau}^c$, we can evaluate the influence of different budgets. Following the formulation in \cite{zhou2023direct}, we draw the curve and use incremental cost as the X-axis and incremental response as the Y-axis. 
Similar to the uplift curve, if a good model generates the score, then the curve should be above the benchmark line. This means a better model would select samples to achieve higher incremental value for the same incremental cost level. The results are shown in Figure~\ref{fig:cost_curve}.
Corresponding to the MT-AUUC results in Table~\ref{tab:overall2}, our E$^3$IR has the best performance of the cost curve on both multi-treatment datasets.
Moreover, we also test the EOM of different approaches when given different budgets; we show the results in Figure~\ref{fig:eom}. 
It is still clear that our E$^3$IR can always help the platform obtain much more profits under different budgets.
\begin{figure}[htbp]
    \centering
    \subfigure[Hillstorm]{\includegraphics[width=0.48\linewidth]{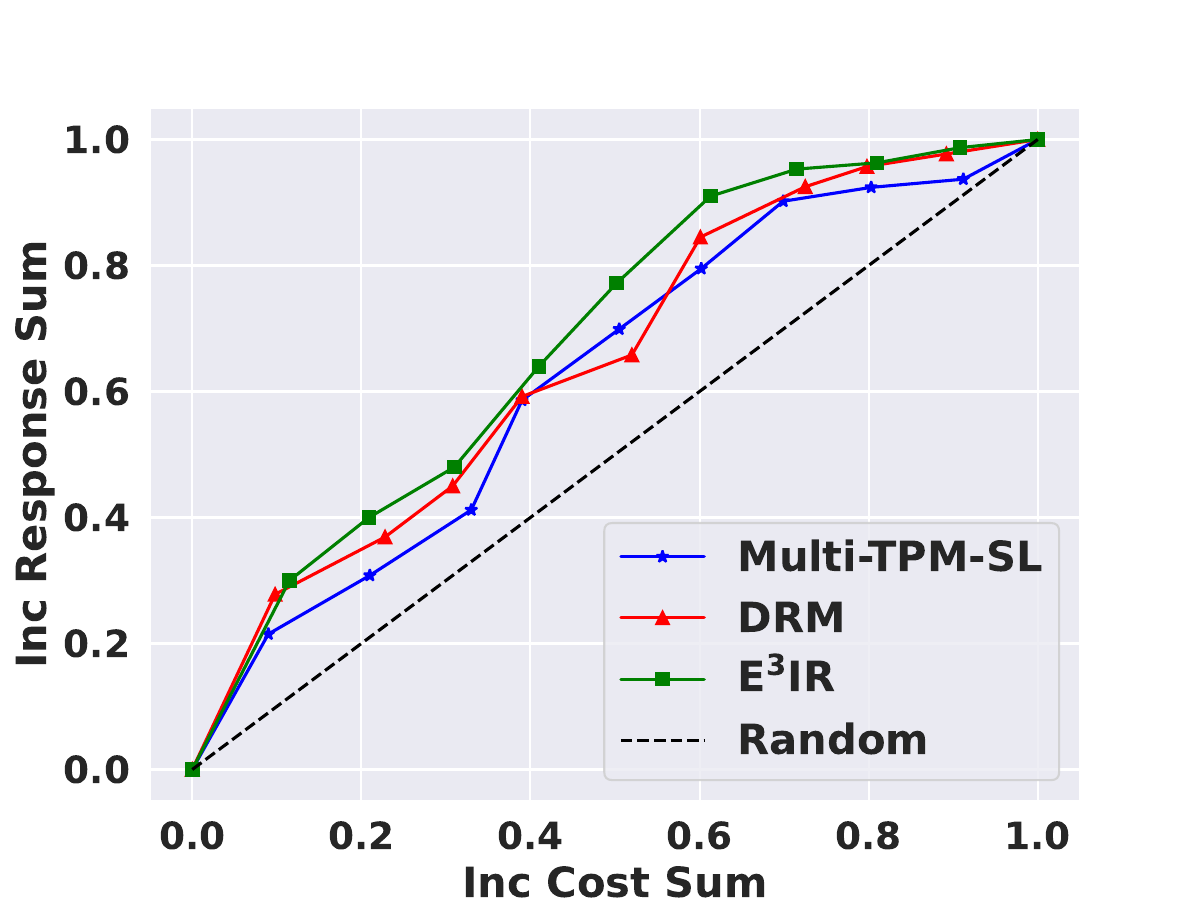}}\quad
    \subfigure[Production]{\includegraphics[width=0.48\linewidth]{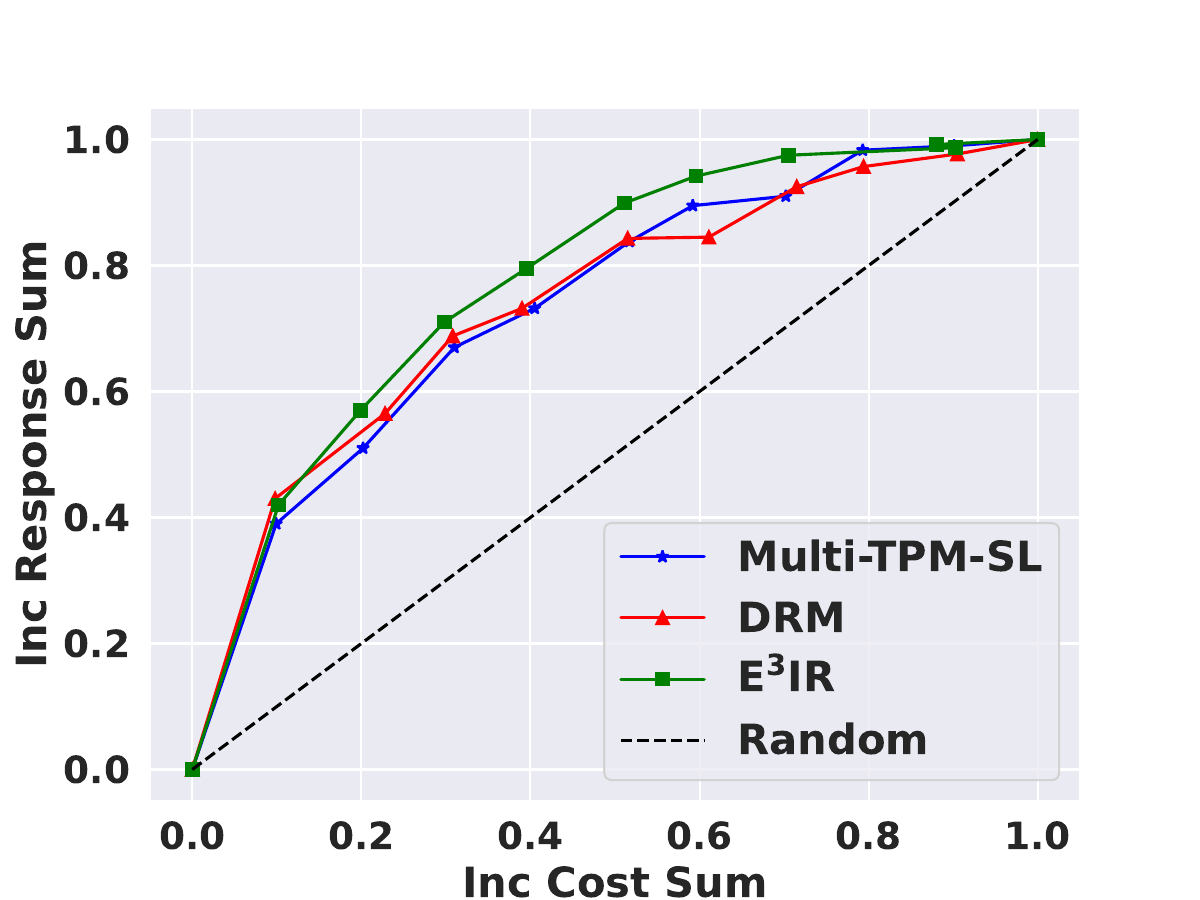}}
    \caption{Cost curve comparison between the cost-aware baselines and our E$^3$IR on the multi-treatment datasets.}
    \Description[<Figure 5. Fully described in the text.>]{<A full description of Figure 5 can be found in Section 5.6.>}
    \label{fig:cost_curve}
\end{figure}

\begin{figure}[htbp]
    \centering
    \subfigure[Hillstorm]{\includegraphics[width=0.48\linewidth]{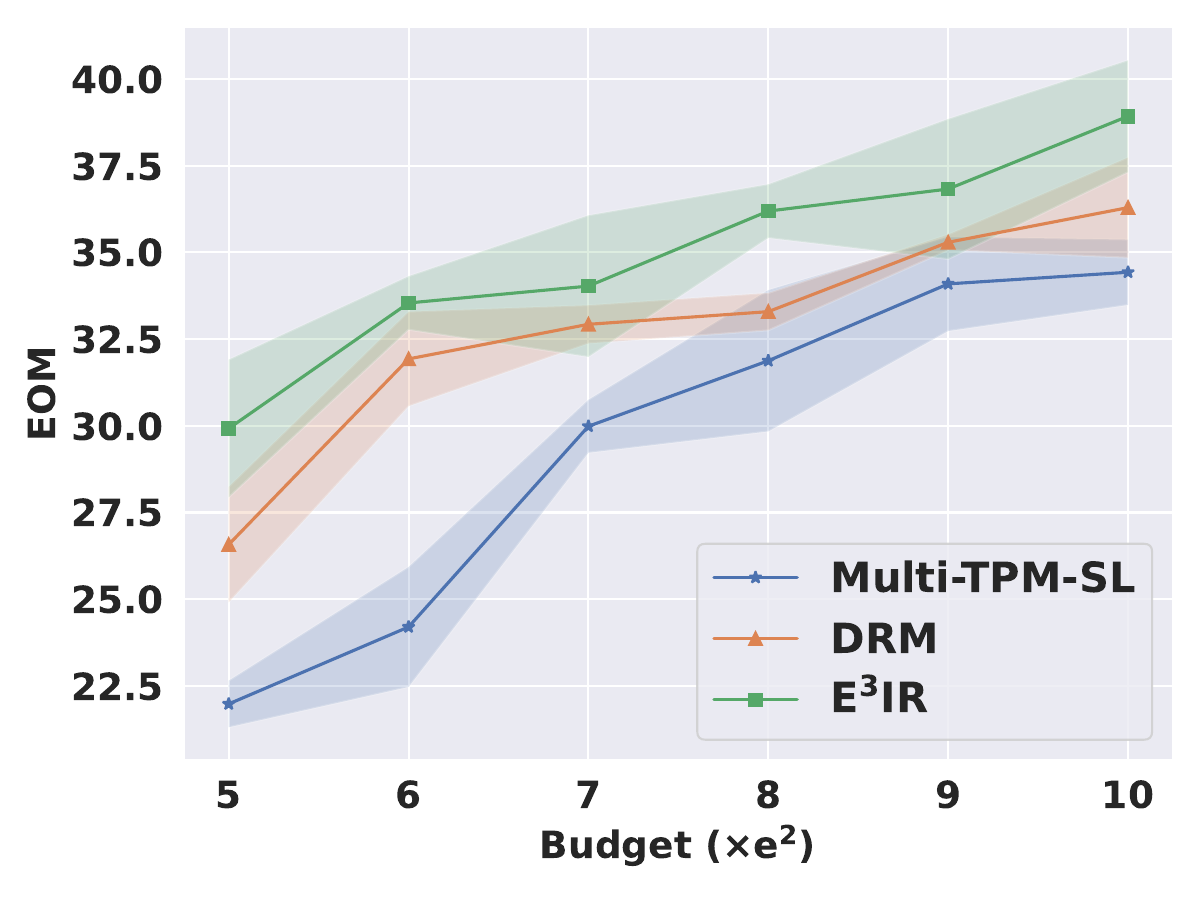}}~~
    \subfigure[Production]{\includegraphics[width=0.48\linewidth]{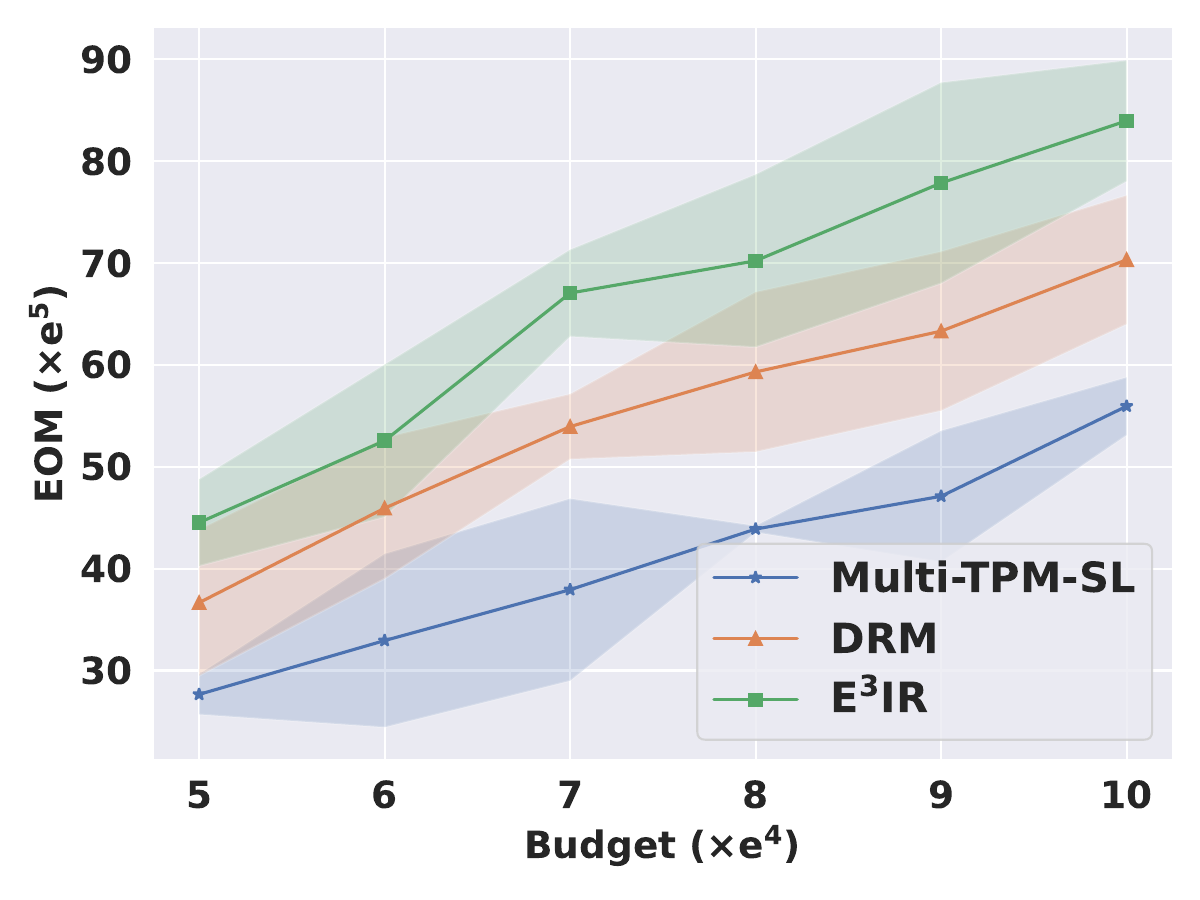}}
    \caption{EOM comparison with the budget changes between the cost-aware baselines and our E$^3$IR on the multi-treatment datasets.}
    \Description[<Figure 6. Fully described in the text.>]{<A full description of Figure 6 can be found in Section 5.6.>}
    \label{fig:eom}
\end{figure}
\section{Conclusion}
In this paper, we formulate online marketing with a specific budget constraint as a ``predict + optimize'' problem. To solve it, we propose an end-to-end optimization method E$^{3}$IR, which consists of two customized modules.  Firstly, we incorporate marketing domain knowledge into the uplift prediction module, enabling the acquisition of a monotonic and smooth user response curve. Additionally, we utilize the differentiable optimization of the ILP problem to reduce the performance gap in two-stage methods. Extensive experiments conducted on binary treatment and multi-treatment datasets demonstrate the superiority of our method across various metrics. Future work will concentrate on reducing the time complexity of the differentiable allocation module and applying this approach to complex real-world marketing scenarios.

\begin{acks}
We thank the support of the National Natural Science Foundation of China (No.62302310).
\end{acks}
\bibliographystyle{ACM-Reference-Format}
\bibliography{sample-base}

\end{document}